\documentclass[11pt]{article}
\usepackage{mathrsfs}
\usepackage{amsmath}
\usepackage{amsfonts}
\usepackage{amssymb}
\usepackage{latexsym}
\usepackage[dvips]{graphicx}
\usepackage{hyperref}
\newtheorem{theorem}{\bf Theorem}[section]
\newtheorem{lemma}[theorem]{\bf Lemma}

\newtheorem{corollary}[theorem]{\bf Corollary}

\newenvironment{proof}{\noindent{\em Proof:}}{\quad \hfill$\Box$\vspace{2ex}}

\def \and {\, \mbox{\rm and}\, }

\def \sgn {\,{\rm sgn}\,}

\def \Re {\,{\rm Re}\,}

\makeatletter

\newcommand{\Rmnum}[1]{\expandafter\@slowromancap\romannumeral #1@}
\makeatother
\begin{document}
\title{\bf A Note on the High-dimensional Sparse Fourier Transform in the Continuous Setting}
\author{Liang Chen\thanks{Department of Mathematics, Jiujiang University, Jiujiang, 332000, China. E-mail address: {\it chenliang3@mail2.sysu.edu.cn}. Support in part by Natural Science Foundation of China under grants 11971490}}
\date{}
\maketitle

%
%

\begin{abstract}
In this paper, we theoretically propose a new hashing scheme to establish the sparse Fourier transform in high-dimension space. The estimation of the algorithm complexity shows that this sparse Fourier transform can overcome the curse of dimensionality. To the best of our knowledge, this is the first polynomial-time algorithm to recover the high-dimensional continuous frequencies.


\end{abstract}

\section{Introduction}

The sparse Fourier transform (SFT) has received continuous attention from applied mathematics (\cite{2010Combinatorial,iwen2013improved,2016Sparse,2017A,bittens2019deterministic,choi2020sparse}), signal processing (\cite{hsieh2013sparse,2014Recent,2014Sparse,pawar2017ffast,2017The,wang2018multidimensional}), and theoretical computer science communities (\cite{kushilevitz1993learning,gilbert2002near,gilbert2005improved,hassanieh2012simple,kapralov2016sparse,kapralov2019dimension}) over the last two decades. Since the sample complexity and runtime complexity of the SFT are mainly affected by the sparsity, and less affected by the bandwidth, the SFT has great advantages in signal processing. Most of the relevant works deals with discrete case where the frequencies are on the grid. Under such condition, the SFT can overcome the curse of dimensionality (\cite{kammerer2017high,kapralov2019dimension,choi2019multiscale,choi2020high,choi2020sparse}). However, this condition that the frequencies are on the grid is so strong. It is natural for people to consider the case that the frequencies are in a continuous region. This has led researchers to establish the sparse Fourier transform in the one-dimensional continuous setting (\cite{boufounos2015s,price2015robust,chen2016fourier}). Recently, \cite{jin2020robust} initiates the study on the SFT in the high-dimensional continuous setting. Unfortunately, their SFT method is still subject to the curse of dimensionality, namely, its runtime complexity is greater than $2^{O(d)}$ ($d$ stands for the dimension).

In this paper, we present a new hashing scheme to transform the high-dimensional SFT into the one-dimensional SFT. The computational complexity of this algorithm is polynomial, which means that the algorithm can break the curse of dimensionality. To the best of our knowledge, this is the first polynomial-time algorithm to recover frequencies in the high-dimensional continuous setting.

Formally, we consider the signal $f$ of the following form

\begin{equation}
f(t)=\sum_{j=1}^{k}f_{j}(t)\triangleq \sum_{j=1}^{k}a_{j}\exp(2\pi i w_{j} \cdot t),
\end{equation} where $t \in \mathbb{R}^{d}, w_{j}\in [-M, M]^d$, $a_{j}\in \mathbb{C}$, $  0<A{'}\le |a_{j}|\le A$ for all $j=1,2,\dots,k$ and $\min_{1\le i<j\le k}|w_{i}-w_{j}|>\eta $.

Our goal is to recover $w_{j},a_{j},j=1,2,\dots ,k$.
\section{Main Result}

We first give some symbols and notions. Let $\sgn(x)=1,\text{when}\quad x\ge 0;$ otherwise, $\sgn(x)=-1$. Let $J^{*}$ denote the transpose of the matrix $J$. Denote for each $t \in \mathbb{R}$ by $\lfloor t\rfloor$ the largest integer not bigger than $t$.
Let $$\Gamma_{1} \triangleq \{(\frac{x_{1}}{F},\dots ,\frac{x_{d}}{F}):x_{i}\in \mathbb{Z}\cap [\frac{-TF}{2},\frac{TF}{2})\},\quad \Gamma_{2} \triangleq \{(\frac{\xi_{1}}{T},\dots ,\frac{\xi_{d}}{T}):\xi_{i}\in \mathbb{Z}\cap [\frac{-TF}{2},\frac{TF}{2})\}.$$
The discrete Fourier transform of the function $g$ takes the following form
$$\mathcal{F}[g](\xi)= \frac{1}{(\sqrt{T}F)^d}\sum_{x\in \Gamma_{1} }g(x)\exp(-2\pi i x \cdot \xi),\quad \xi\in\Gamma_{2},$$
and the inverse discrete Fourier transform of the function $g$ is defined by
$$\mathcal{F}^{-1}[g](x)= \frac{1}{(\sqrt{T})^d}\sum_{\xi\in \Gamma_{2} }g(\xi)\exp(2\pi i x \cdot \xi),\quad  x\in\Gamma_{1}.$$
For $x,y\in \Gamma_{1}$, $x\pm y\triangleq x\pm y (mod T)\in \Gamma_{1}$. For $x,y\in \Gamma_{2}$, $x\pm y\triangleq x\pm y (mod F)\in \Gamma_{2}$.

The ``bucket" is defined by
 \begin{equation}
 B_{j}\triangleq \{(\xi_{1}/T, \dots ,\xi_{d}/T)\in \Gamma_{2}  : \xi_{d}\in [ \frac{TF(j-1)}{s}-\frac{FT}{2}, \frac{TFj}{s}-\frac{FT}{2})\}, \quad 1\le j\le s.
\end{equation}

The hashing transform is defined by
 $
H(\xi)\triangleq
h(\xi)+(0,\dots ,0,b/T)^{*}
, $  where $$h(\xi)=
\begin{bmatrix}
\mathbf{I}_{d-1} & 0   \\
V_{h} & h_{d}\\
\end{bmatrix}\xi,$$ here $\mathbf{I}_{d-1}$ is the identity matrix of order $d-1$, $V_{h}=(h_{1},\dots,h_{d-1})$ and $\{h_{1},\dots,h_{d}\}$ are independent draws from  the uniform distribution on the set $\{2n+1:n\in \mathbb{Z}\}\cap[0,F/\eta],$ where $T>1/\eta$. The random variable $b$ obeys the uniform distribution on the set $\mathbb{Z} \cap \{[0,\frac{TF}{s})\}$.

In this paper, we assume that $T,F,s$ are powers of $2$ and $1<s<F$. Next, we give a key lemma.
\begin{lemma}\label{22} Suppose $0<\delta<\frac{1}{2}$, $0<\epsilon<\min \{1,\eta,\frac{A^{'}}{4},\frac{1}{4A^2}\}$, $s=\mathcal{O}(\frac{\sqrt{d} k^{2}}{\delta} )$, $T=\mathcal{O}(k^{4}d^{5/2}/((\frac{\epsilon \delta}{ds})^{2}\eta\delta^{2}))$, $F=\mathcal{O}(k^{2}M/\delta)$. With probability at least $1-\delta/4$ over the randomness of $\{h_{1},\dots,h_{d},b\}$, for each $j\in\{1,\dots,s\}$,  
either \begin{equation}
 I_{j,1}\triangleq \frac{1}{(TF)^d}\sum_{x \in \Gamma_{1}}
 \big|e^{\frac{2\pi i bx_{d}}{TF}}\mathcal{F}^{-1}[\mathcal{X}_{j}\cdot \mathcal{F}[f_{H}]]((h^{-1})^{*} x ) \big|^2  \le (\frac{A\epsilon \delta}{3ds})^2.
 \end{equation}
or there exists a unique $j_{k}\in\{1,2,\dots,k\}$ such that
 \begin{equation}
I_{j,2}\triangleq \frac{1}{(TF)^d}\sum_{x \in \Gamma_{1}}
 \big|f_{j_{k}}(x)-e^{\frac{2\pi i bx_{d}}{TF}}\mathcal{F}^{-1}[\mathcal{X}_{j}\cdot \mathcal{F}[f_{H}]]((h^{-1})^{*} x) \big|^2  \le (\frac{A\epsilon \delta}{3ds})^2.
 \end{equation}
 Where $$f_{H}(x)\triangleq f((h^{*}x)\exp(-2\pi i (h^{*}x\cdot (0,\dots,0,b/T)^{*} ),$$ and
 $$\mathcal{X}_{j}(\xi)\triangleq\left\{
\begin{array}{rcl}
1       &      & {\xi  \in  B_j}\\
0     &      & {\xi  \notin  B_j}
\end{array} \right.  , \quad j=1,2,\dots,s.$$
\end{lemma}
By Markov's inequality, we have a corollary.
\begin{corollary}\label{212} All parameters are set as Lemma \ref{22}. For any $l\in\{1,2,\dots,d\}$,
choosing the random vector $X_{l}\triangleq(x_{1},\dots,x_{l-1},x_{l+1},\dots,x_{d})$ which obeys the uniform distribution on the set $\mathbb{Z}^{d-1}\cap [\frac{-TF}{2},\frac{TF}{2})^{d-1}$, then with probability at least $1-\delta/(4ds)$ over the randomness of $(x_{1},\dots,x_{l-1},x_{l+1},\dots,x_{d})$, we have
$$\sum_{x_{l}\in \mathbb{Z}\cap [\frac{-TF}{2},\frac{TF}{2}) } \frac{1}{TF}
 \big|e^{\frac{2\pi i bx_{d}}{TF}}\mathcal{F}^{-1}[\mathcal{X}_{j}\cdot \mathcal{F}[f_{H}]](h^{*}(x)) \big|^2 \le  dsI_{j,1}/\delta$$
or
$$\sum_{x_{l}\in \mathbb{Z}\cap [\frac{-TF}{2},\frac{TF}{2}) }\frac{1}{TF}\big|f_{j_{k}}(x)-e^{\frac{2\pi i bx_{d}}{TF}}\mathcal{F}^{-1}[\mathcal{X}_{j}\cdot \mathcal{F}[f_{H}]](h^{*}(x)) \big|^2  \le  dsI_{j,2}/\delta .$$
Where  $ x=(x_{1}/F, \dots ,x_{l-1}/F,x_{l}/F,x_{l+1}/F,\dots , x_{d}/F)^{*}   . $
\end{corollary}

We explain Lemma \ref{22} at a high level. It is well known that for the discrete Fourier transform, enlarge the sampling interval can enhance the frequency concentration (the frequency resolution) of each component $f_{j}$ (see Lemma \ref{44}), and the frequency gap between different components is greater than a constant value $\eta$. Therefore, with the expanding of the sampling interval, $\eta$ will be greater than the radius of the support set (the main part of the energy) for each signal component in the frequency domain. And the hashing transform we defined can transform the gap between frequencies to the gap between the last coordinates of the frequencies. Thus, this hashing transform can keep the main frequencies of the same component into the same bucket, and different components can be isolated into different buckets.

 When the different components of the signal are isolated, we only need to recover the frequency and amplitude in each bucket. By Lemma \ref{22}, there are only two cases of such amplitude and frequency, one is that they are close to the amplitude and frequency of a certain component (up to the hashing transformation), and the other is that the amplitude is less than $O(\epsilon)$.

 Specifically, for $1\le j\le s$, $1\le l\le d$, set $$g_{j,l}(t)=\exp(2\pi i \delta _{d,l}b\lfloor tF\rfloor/(TF))\mathcal{F}^{-1}[\mathcal{X}_{j}\cdot \mathcal{F}[f_{H}]](h^{*}(x_{l,t})), $$
 where $\delta_{d,1}=1,\text{when}\quad l=d;$ otherwise, $\delta_{d,l}=0$ and $$x_{l,t}=(x_{1},\dots,x_{l-1},\lfloor tF\rfloor /F,x_{l+1},\dots,x_{d}).$$ Here $(x_{1},\dots,x_{l-1},x_{l+1},\dots,x_{d})$ is the random vector obeys the uniform distribution on the set $\mathbb{Z}^{d-1}\cap [\frac{-TF}{2},\frac{TF}{2})^{d-1}$. The sample value of the function $g_{j,l}$ is calculated by the method in Lemma \ref{33}. Let $F\ge \mathcal{O}(\sqrt{d}M/\epsilon)$, by Lemma \ref{22} and Corollary \ref{212}, either there is a signal component $f_{j_{k}}$ with frequency $w_{j_{k}}$ such that $$\frac{1}{T}\int_{-T/2<t<T/2} |g_{j,l}(t)-a_{j_{k}}\theta_{j_{k},l} e^{2\pi i w_{j_{k},l} t}|^2 dt \le \mathcal{O}(A^2\epsilon^2),$$ or $$ \frac{1}{T}\int_{-T/2<t<T/2} |g_{j,l}(t)|^2 \le \mathcal{O}(A^2\epsilon^2),$$ where $w_{j_{k},l}$ denotes $l$-th component of the vector $w_{j_{k}}$ and $$\theta_{j_{k},l}=\exp(\frac{2\pi i}{F}(w_{j_{k},1}x_{1} +\dots+w_{j_{k},l-1}x_{l-1} +w_{j_{k},l+1}x_{l+1}+\dots w_{j_{k},d}x_{d} )).$$

We have the following {\it algorithm}:
For each $j\in \{1,2,\dots,s\}$, we first use the algorithm in \cite{price2015robust} to recover the amplitude and frequency of $g_{j,1}(t)$. If the absolute value of the output amplitude less than $A^{'}/2$, we turn to recover the amplitude and frequency of $g_{j+1,1}(t)$; otherwise, we keep $ w_{j_{k},1}^{o}$ and continue to use the algorithm in \cite{price2015robust} to recover the frequencies of $g_{j,2}(t),\dots,g_{j,d}(t)$. Finally, $\{w_{j_{k},1}^{o},\dots ,w_{j_{k},l}^{o},\dots ,w_{j_{k},d}^{o}\}$ can be recovered element-wisely.

 {\it Note}: From Theorem 1.1 in \cite{price2015robust}, it takes at most $\mathcal{O}(ds\log(TF)\log(ds/(\delta\epsilon))/\delta)$ samples of the function $g_{j,l}$ to recover the frequency $w_{j_{k},l}$ (up to the accuracy $A^2\epsilon/(TA') $) with probability at least $1-\delta/(4ds)$. By Lemma \ref{33}, in order to get $\mathcal{O}(ds\log(TF)\log(ds/(\delta\epsilon)))$ samples (up to the accuracy $A^{2}\epsilon$) of $g_{j,l}$ with probability at least $1-\delta$, we need a total of   $$\mathcal{O}( k^{2}ds\ln^{2}(TF+1) \log(TF)\log(TFds /(\delta\epsilon))\ln^{2}(1/\delta)/(\delta\epsilon^2) )$$ samples and running time.
Therefore, we have the following result.

\begin{theorem}
All parameters are set as Lemma \ref{22}, besides, $F\ge \mathcal{O}(\sqrt{d}M/\epsilon)$. The above algorithm can output $\{  w_{1}^{o} ,\dots, w_{k}^{o} \}$ such that $$|w_{j}^{o}-w_{j}|<\mathcal{O}(\frac{\epsilon A^{2}\sqrt{d}}{A{'}T})<\mathcal{O}(\frac{A^2\epsilon^3/A{'}}{\sqrt{d}k^4}),  \quad i=1,2,\dots,k.$$ holds with probability at least $1-\delta$ over the randomness of $\{h_{1},\dots,h_{d},b\}$, $\{X_{l}\}_{l=1}^{d}$ (see Corollary \ref{212}), $\mathcal{T}$ (see Lemma \ref{33}). The runtime complexity and sample complexity are at most $$\mathcal{O}( k^{2}d^2s^2\ln^{2}(TF+1) \log(TF)\log(TFds /(\delta\epsilon))\ln^{2}(1/\delta)/(\epsilon^2 \delta)).$$

\end{theorem}
The success probability can be boosted to $1$ by repeatedly restarting (as indicated in (\cite{gilbert2002near,price2015robust})).

Since $|w_{j}^{o}-w_{j}|<\mathcal{O}(\frac{A^2\epsilon^3/A{'}}{\sqrt{d}k^4})$ and $\epsilon<\eta$, we have $$ a_{j_{ }}=(\epsilon^{2}/k)^d\int_{t\in[0,k/\epsilon^2]^d}f(t)e^{-2\pi i w_{j}^{o}\cdot t }dt+\mathcal{O}( A^2\epsilon/(A{'}k^3)).$$ We can use the Monte Carlo method to compute the integral $(\epsilon^2/k)^d\int_{t\in[0,k/\epsilon^2]^d}f(t)e^{-2\pi i w_{j}^{o}\cdot t }dt$, by Hoeffding's inequality, with probability at least $1-\delta$, the amplitude $ a_{j}$ can be recovered (up to $A^{2}\epsilon/A{'}$) with $\mathcal{O}(\log(1/\delta)k^{2}A{'}^{2}/\epsilon^2 )$ random samples.
\section{Proof of Lemma \ref{22}}
 To prove Lemma \ref{22}, we need some technical lemmas.
\begin{lemma}\label{44} Let $0<\beta <F/2$, for each $f_{j}$, the following inequalities $$\frac{1}{(TF)^d}\sum_{x \in \Gamma_{1}} \big|f_{j}(x)-\mathcal{F}^{-1}[\mathcal{X}_{j,\beta}^{'}\cdot \mathcal{F}(f_{j})] (x)  \big|^2 \le \mathcal{O}(\frac{d^{3/2}A^2}{T\beta}),$$ holds for any $j\in\{1,2,\dots,k\}.$ Where
 $$\mathcal{X}_{j,\beta}^{'}(\xi)\triangleq\left\{
\begin{array}{rcl}
1       &      & {|\xi -w_{j }|\le \beta/2  }\\
0     &      & {otherwise}
\end{array} \right.  ,   j=1,\dots,k.$$

\end{lemma}
 \begin{proof} For any $\alpha \in [0,1/T)$ and any $n\in\mathbb{Z}\cap [-FT/2,FT/2)$, we have
\begin{equation}\label{GD}
\frac{1}{TF}  \bigg|\sum_{j\in \mathbb{Z}\cap [\frac{-TF}{2},\frac{TF}{2})}\exp(2\pi i \alpha j/F)\exp\bigg( \bigg(-\frac{2\pi ij}{F}\bigg)\bigg(\frac{n}{T}\bigg)\bigg)\bigg|\le \frac{2}{TF|1-\exp(2\pi i (1/F)(\alpha -\frac{n}{T}))|}
\le \mathcal{O}(\frac{1}{n}).
\end{equation}
Since $|a_{j}|\le A$, 
 by (\ref{GD}), using Parseval's identity, we have
\begin{equation}
 \frac{1}{T^d}\sum_{\xi \in \Gamma_{2}}|\mathcal{X}_{j,\beta}^{'}(\xi )\cdot \mathcal{F}[f_{j}](\xi )- \mathcal{F} (f_{j})(\xi)|^2\le \mathcal{O}(\frac{d^{3/2}A^2}{T\beta}),
 \end{equation}
which completes the proof.
\end{proof}
\begin{lemma}\label{see}
Let $s=\mathcal{O}(\sqrt{d}k^{2}/\delta)$, $F=\mathcal{O}(k^{2}M/\delta)$. For any $w\in[-M,M]^d$ with $|w|\ge \eta$, $|(h(w))_{d}|>2F/s $ holds with probability at least $1-\delta/k^2$ over the randomness of ${h_{1},h_{2},\dots,h_{d}}$, where $(h(w))_{d}$ denotes $d$-th component of the vector $h(w)$.
\end{lemma}
\begin{proof}  Without loss of generality, let's assume that the $d$-th component $w_{d}$ of $w$ is greater than $\eta/\sqrt{d}$. Fix arbitrary ${h_{1},\dots,h_{d-1}}$, since $\eta/\sqrt{d}\le|w_{d}|\le 2M$, with probability at most $\delta/k^2$ (no greater than the ratio of the width of the set $(-2F/s,2F/s)$ to $F/6$) over the randomness of ${h_{d}}$, $\sum_{i=1}^{d}h_{i}w_{i}\in(-2F/s,2F/s)$, which complete the proof.
 \end{proof}

From Lemma \ref{see}, we have a direct consequence.
\begin{corollary}\label{c1}
All parameters are set as above lemma. Suppose the vectors ${w_1,\dots,w_k}$ satisfy $$\min_{i\neq j}|w_{i}-w_{j}|>\eta, \quad\max_{1\le i\le k}|w_{i}|\le M, $$  then $$\min_{1\le i<j\le k} |(h(w_{i}-w_{j}))_{d}|>2F/s$$ holds with probability at least $1-\delta$ over the randomness of ${h_{1},h_{2},\dots,h_{d}}$.
\end{corollary}

 Suppose $F\beta/\eta\le F\delta/(\sqrt{d}ks)$. For any $\xi_{1},\xi_{2}\in \phi_{j,\beta}\triangleq \{\xi:|\xi-w_{j}|\le \beta/2 \}$, we have $|(h(\xi_{1})-h(\xi_{2}))_{d}|\le \delta F/ks< F/s$. However, this does not mean that there is a bucket $B_{j_{s}}$ such that $h(\phi_{j,\beta})\subset B_{j_{s}}$.  It is possible that $h(\phi_{j,\beta})$ intersects on the boundary of some bucket. Obviously, after adding random translation, with a certain probability (the ratio of the total ``width" of the sets $\phi_{j,\beta},j=1,2,\dots,k$ to the ``width" of the bucket), there exists a bucket $B_{j_{s}}$  such that $H(\phi_{j,\beta})$ is completely inside it, that is, the following conclusion holds.
 \begin{lemma}\label{last}
  Suppose $F\beta/\eta\le F\delta/(\sqrt{d}ks)$, with probability at least $1-\delta$ over the randomness of $b$, there exists a bucket $B_{j_{s}}$ such that $H(\phi_{j,\beta})\subset B_{j_{s}}$ for j=1,2,\dots,k.
 \end{lemma}

Choosing $s=\mathcal{O}(\sqrt{d}(k^{2}/\delta))$, $\beta=\eta \delta/(\sqrt{d}ks)$, $T=\mathcal{O}(k^{4}d^{5/2}/((\frac{\epsilon \delta}{ds})^{2}\eta\delta^{2}))$ and $F=\mathcal{O}(k^{2}M/\delta)$. Since $\mathcal{F}[f_{H}](\xi)=\mathcal{F}[f](H^{-1}\xi)$, combining Lemma \ref{see}, Corollary \ref{c1}, Lemma \ref{last} and Lemma \ref{44} proves Lemma \ref{22}.

\section{Proof of Lemma \ref{33}}
 The following lemma reduces the computational complexity of the convolution operation.
\begin{lemma}\label{33}For any $0<\epsilon<1/2$, $1\le j\le k$,
suppose $$N=\mathcal{O}(\frac {k^{2}A^{2}\ln^{2}(TF)\ln^{2}(1/\delta)}{\epsilon^2}),$$  then with probability at least $1-\delta/4$ over the randomness $\mathcal{T}$, the following
inequality holds
\begin{equation}
 \sup_{x\in \Gamma_{1}}\bigg|\mathcal{F}^{-1}[\mathcal{X}_{j_{s}}\cdot \mathcal{F}[f_{H}]](x)-\frac{1}{N} \sum_{i=1}^{N}f_{H}(x-\big(0,\dots , 0,\frac { \lfloor  \sgn (t)((\frac{TF}{2} +1)^{2|t|}-1 ) \rfloor }{F}\big)^{*})v(t_{i})\bigg| \le \epsilon,
\end{equation}
where  $\mathcal{T}\triangleq \{t_{1},\dots,t_{N}\}$ are independent draws from the uniform distribution on $(-1/2,1/2)$,   $$v(t)=   2\ln (\frac{TF}{2} +1)(\frac{TF}{2}+1)^{2t} v_{2}(\lfloor  \sgn (t)((\frac{TF}{2}+1)^{2|t|}-1)  \rfloor), $$ and

 \begin{equation}
  v_{2}(y)\triangleq\left\{
\begin{array}{rcl}
\frac{\exp(-\pi i y)(\exp(\frac {2\pi i y(j-1)}{s})-
\exp(\frac{2\pi i  yj}{s} ))}{TF(1-\exp(2\pi i y/TF ))}       &      & {y\neq 0}\\
1/s    &      & {y=0}
\end{array} \right ..
\end{equation}

\end{lemma}
 \begin{proof} Observe that
$$  \sum_{w\in \mathbb{Z}\cap [-TF/2,TF/2)}\exp(\frac {2\pi i nw}{TF})=0 \quad \text{for} \quad n\in \mathbb{N}^{+}\cup \mathbb{N}^{-},$$
then
\begin{equation}\label{est}
\begin{array}{ll}
&\mathcal{F}^{-1}[\mathcal{X}_{j_{s}}\cdot \mathcal{F}[f_{H}]](x)=\sum_{y\in \Gamma_{1}}\frac {f_{H}(x-y)\mathcal{F}^{-1}(\mathcal{X}_{j})(y)}{(\sqrt{T}F)^d}\\&=\sum_{z\in \mathbb{Z}\cap [\frac{-TF}{2},\frac{TF}{2})}f_{H}(x-(0,\dots , 0,\frac {z}{F})^{*})  v_{2}(z)\\
&= \int_{0}^{TF/2}f_{H}(x-(0,\dots ,0,\frac { \lfloor z\rfloor }{F})^{*}) (z+1)2\ln (1+TF/2 )v_{2}(z)d\big(\frac{\ln (z+1)}{2\ln (1+TF/2)}\big)\\
&+ \int_{-TF/2}^{0}f_{H}(x-(0,\dots ,0,\frac { \lfloor z\rfloor }{F})^{*}) (-z+1)2\ln (1+TF/2 )v_{2}(z)d\big(\frac{\ln (-z+1)}{2\ln (1+TF/2)}\big)\\
&= \int_{-1/2}^{1/2} f_{H}(x-(0,\dots , 0,\frac { \lfloor  \sgn (t)((\frac{TF}{2} +1)^{2|t|}-1 ) \rfloor }{F})^{*})v(t)dt,
\end{array}
\end{equation}
where $\sup_{t\in [0,1)}|v(t)|\le  \mathcal{O}( \ln(TF))$.

Next, we consider the empirical Rademacher complexity (\cite{shalev2014understanding,allen2019learning}) of the following function spaces
$$Q_{w,u,C_{0}}\triangleq \{q_{x}(t)\triangleq u(t)\sin(2\pi w (\frac {x -\lfloor  \sgn (t)((\frac{TF}{2} +1)^{2|t|}-1 ) \rfloor}{F} )  \}, $$
and
  $$Q_{w,u,C_{0}}^{'}\triangleq \{q_{x}^{'}(t)\triangleq u(t)\cos(2\pi w (\frac {x -\lfloor  \sgn (t)((\frac{TF}{2} +1)^{2|t|}-1 ) \rfloor}{F} )  \},$$
where $w,x,t\in \mathbb{R}$ and $u(t)$ is any real-valued function satisfying $|u| \le C_{0}$.
Let $\widehat{\Re}(\mathcal{T} ; Q_{w,u,C_{0}})$ denote the empirical Rademacher complexity for $Q_{w,u,C_{0}}$ and $\mathcal{T}$, then
\begin{equation}\label{88}
\begin{array}{ll}
&\widehat{\Re}(\mathcal{T} ; Q_{w,u,C_{0}}) = \mathbb{E}_{\xi \sim\{\pm 1\}^{N}}\left[\sup _{x \in \mathbb{R}}  \sum_{i=1}^{N} \frac{\xi_{i} q_{x}\left(t_{i}\right)}{N}\right]\\&\le\mathbb{E}_{\xi \sim\{\pm 1\}^{N}}\bigg [\sup _{(z_{1},z_{2} )\in [-1,1]^2}\sum_{i=1}^{N} \frac{\xi_{i}(z_{1}y_{i,1}-z_{2}y_{i,2})}{N}\bigg]\le \frac{2C_{0}}{\sqrt{N}},
\end{array}
\end{equation}
 where 
 $$y_{i,1}=\cos(\frac{2\pi w\lfloor  \sgn (t_{i})((\frac{TF}{2} +1)^{2|t_{i}|}-1 ) \rfloor}{F})u(t_{i}),\quad y_{i,2}=\sin(\frac{2\pi w\lfloor  \sgn (t_{i})((\frac{TF}{2} +1)^{2|t_{i}|}-1 ) \rfloor}{F})u(t_{i}),$$ the last inequality in (\ref{88}) is obtained by Lemma 26.10 in \cite{shalev2014understanding}. Similarly, we have $\widehat{\Re}(\mathcal{T} ; Q_{w,u,C_{0}}^{'})\le \frac{2C_{0}}{\sqrt{N}}$.

Combining the last equality in (\ref{est}) and Lemma A.10 in \cite{allen2019learning} proves Lemma \ref{33}.

\end{proof}

\section{Conclusion}




We do not consider the case of noise in this paper,  so it is an issue that need to be discussed later. In addition, we can also discuss the reconstruction of the signals without frequency gap just like one-dimensional case \cite{chen2016fourier}. In terms of application, it is worthy to optimize the algorithm and the complexity estimation to make it more applicable to practice.

\end{document}